\documentclass[12pt, a4paper, twoside]{article}
\usepackage{amsmath,amsfonts,amssymb,amsthm,amscd}
\usepackage{epsfig}
\usepackage{hyperref}
\usepackage{enumerate} % allows different counter for items (e.g. (a) instead of 1.)
\newtheorem{theorem}{Theorem}

\newtheorem{lemma}[theorem]{Lemma}
\newtheorem{corollary}[theorem]{Corollary}
\newtheorem{conjecture}{Conjecture}
\newtheorem{observation}[theorem]{Observation}

\def\l{\left}
\def\r{\right}

\def\bsd#1{{\underline{#1}}}
\def\bsu#1{{\overline{#1}}}
\def\flipped#1{{\overline{#1}}}

\begin{document}

\title{Average number of flips in pancake sorting \\[20pt]}

\author{Josef Cibulka\\
\small{Department of Applied Mathematics, Charles University,} \\
\small{Malostransk\'e n\'am.~25, 118~ 00 Prague, Czech Republic. }\\
\small{\it cibulka@kam.mff.cuni.cz}
\thanks{Work on this paper was supported by the project 1M0545 
of the Ministry of Education of the Czech Republic
and by the Czech Science Foundation under the contract no.\ 201/09/H057.
The access to the METACentrum computing facilities provided 
under the research intent MSM6383917201 is highly appreciated.}
}
\date{}
\maketitle

\begin{abstract}
We are given a stack of pancakes of different sizes and the only
allowed operation is to take several pancakes from top and flip them.
The unburnt version requires the pancakes to be sorted by their sizes at 
the end, while in the burnt version they additionally 
need to be oriented burnt-side down. 
We present an algorithm with the average number of flips, needed to sort a stack of $n$ 
burnt pancakes, equal to $7n/4 + O(1)$ and a randomized algorithm for the unburnt 
version with at most $17n/12 + O(1)$ flips on average.

In addition, we show that in the burnt version, the average number of flips of any algorithm is 
at least $n+\Omega(n/\log n)$ and conjecture that some algorithm can reach $n+\Theta(n/\log n)$.

We also slightly increase the lower bound on $g(n)$, the minimum number of flips needed
to sort the worst stack of $n$ burnt pancakes. This bound together with the upper bound found 
by Heydari and Sudborough in 1997 gives the exact number of flips to sort the previously
conjectured worst stack $-I_n$ for $n \equiv 3 \pmod 4$ and $n \geq 15$.

Finally we present exact values of $f(n)$ up to $n=19$ and of $g(n)$ up to $n=17$ and 
disprove a conjecture of Cohen and Blum by showing that the burnt stack $-I_{15}$ is not the worst one 
for $n=15$.
\end{abstract}

\emph{Keywords\/}: Pancake problem, Burnt pancake problem, Permutations, Prefix reversals, Average-case analysis

\section{Introduction}
The pancake problem was first posed in \cite{Dweighter}.
We are given a stack of pancakes each two of which have different sizes and
our aim is to sort them in as few operations as possible to obtain a stack of pancakes with 
sizes increasing from top to bottom. The only allowed sorting operation 
is a "spatula flip", in which a spatula is inserted beneath an arbitrary pancake, 
all pancakes above the spatula are lifted and replaced in reverse order. 

We can see the stack as a
permutation $\pi$. A flip is then a prefix reversal of the permutation.
The set of all permutations on $n$ elements is denoted by $S_n$, $f(\pi)$ is the minimum number 
of flips needed to obtain $(1,2,3,\dots,n)$ from $\pi$ and 
\[
f(n) := \max_{\pi \in S_n}f(\pi).
\]

%$f(n)$ be the maximum of $f(\pi)$ over all permutations on $n$ elements. 

The exact values of $f(n)$ are known for all $n\leq 19$, see 
Table~\ref{table:val} for their list and references.
In general $15\lfloor n/14\rfloor \leq f(n) \leq 18n/11 + O(1)$. 
The upper bound is due to Chitturi et al.~\cite{Chitturi+2008} 
% who improved the upper bound $(5n+5)/3$ of Gates and Papadimitriou~\cite{GatesPapad}
and the lower bound was proved by Heydari and Sudborough~\cite{HeydariSudb}.
These bounds improved the previous bounds $17n/16 \leq f(n) \leq (5n+5)/3$ 
due to Gates and Papadimitriou~\cite{GatesPapad}, where the upper bound was also
independently found by Gy\"{o}ri and Tur\'{a}n~\cite{GyoriTuran}.

%The pancake sorting has applications in parallel processing. The $n$-dimensional pancake 
%network is composed of processors labeled by permutations on $n$ elements. Two processors
%labeled $\pi$ and $\rho$ are directly connected if $\pi$ can be obtained from $\rho$ 
%by a prefix reversal. The largest distance between two processors is called the diameter 
%of the network. The diameter of the  $n$-dimensional pancake network is equal to $f(n)$.

A related problem in which the reversals are not restricted to intervals containing the first 
element received considerable attention in computational biology; see e.\ g.~\cite{Hayes2007}. 

%The number of reversals
%needed to transform a DNA sequence of one species to that of another is used to
%determine the evolutionary distance between the two species. 
%The problem of finding the shortest sequence of reversals between two permutations is 
%NP-complete~\cite{Caprara1997}. But if the elements of the permutations are signed, the problem
%is in P; the fastest known algorithm runs in $O(n\sqrt{n\log(n)})$ time~\cite{TannierSagot2004}. 
%The signed version of the problem for $n=2$ was also solved on a biological computer 
%composed of a large number of bacteria~\cite{Biological2008}.

A variation on the pancake problem is the burnt pancake problem in which pancakes are burnt on one of their sides. 
This time, the aim is not only to sort them by their sizes, but we also require that at the end, they all have
their burnt sides down. Let $C=(\pi,v)$ denote a stack of $n$ burnt pancakes,
where $\pi \in S_n$ is the permutation of the pancakes and $v \in \{0,1\}^n$ is the vector of their orientations
($v_i=0$ if the $i$-th pancake from top is oriented burnt side down). 
Pancake $i$ will be represented by $\bsd i$ if its burnt side is down and $\bsu i$ if up.
Let 
\[
I_n=
\l(
\begin{array}{c}
	\bsd{1} \\ \bsd{2} \\ \vdots \\ \bsd{n}
\end{array}
\r)
\qquad \text{and} \qquad
-I_n=
\l(
\begin{array}{c}
	\bsu{1} \\ \bsu{2} \\ \vdots \\ \bsu{n}
\end{array}
\r).
\]

Let $g(C)$ be the minimum number of flips needed to obtain $I_n$ from $C$ and let
\[
g(n) := \max_{\pi \in S_n, v \in \{0,1\}^n}g((\pi,v)).
\]

%$g(n)$ the maximum of $g(C)$ over all stacks $C$.

Exact values of $g(n)$ are known for all $n\leq 17$, see Table~\ref{table:val}.
In 1979 Gates and Papadimitriou~\cite{GatesPapad} provided the bounds $3n/2-1 \leq g(n) \leq 2n+3$. 
Since then these were improved only slightly by Cohen and Blum~\cite{CohenBlum} to $3n/2 \leq g(n) \leq 2n-2$, 
where the upper bound holds for $n \geq 10$.
The result $g(16)=26$ further improves the upper bound to $2n-6$ for $n\geq 16$.
Cohen and Blum also conjectured that the maximum number of flips is always achieved for the stack $-I_n$.
But we present two counterexamples with $n=15$ in Section~\ref{sec:comp}.
%which is obtained from $I_n$ by changing the orientation of each pancake. 

The stack $-I_n$ can be sorted in $(3(n+1))/2$ flips for 
$n \equiv 3 \pmod 4$ and $n \geq 23$~\cite{HeydariSudb}.
%, therefore $3n/2+O(1)$ would be an upper bound on $f(n)$ under the conjecture. 
In Section~\ref{sec:lb} we present a new formula for
determining a lower bound on the number of flips needed to sort a given stack of burnt pancakes. The highest
value that this formula gives for a stack of $n$ pancakes, is $\lfloor (3(n+1))/2 \rfloor$ for the stack $-I_n$.
These bounds together with the known values of $g(-I_{15})$ and $g(-I_{19})$ give
$g(-I_n)=(3(n+1))/2$ if $n \equiv 3\pmod 4$ and $n \geq 15$.

\begin{table}[ht]
\centering
\begin{tabular}{r rl rl rl}
%\hline
$n$ & $f(n)$ & & $g(n)$ & & $g(-I_n)$ & \\ 
\hline
2 &    1 & \cite{Garey+1977}   &    4 & \cite{CohenBlum}       &    4 & \cite{CohenBlum} \\
3 &    3 & \cite{Garey+1977}   &    6 & \cite{CohenBlum}       &    6 & \cite{CohenBlum} \\
4 &    4 & \cite{Garey+1977}   &    8 & \cite{CohenBlum}       &    8 & \cite{CohenBlum} \\
5 &    5 & \cite{Garey+1977}   &   10 & \cite{CohenBlum}       &   10 & \cite{CohenBlum} \\
6 &    7 & \cite{Garey+1977}   &   12 & \cite{CohenBlum}       &   12 & \cite{CohenBlum} \\
7 &    8 & \cite{Garey+1977}   &   14 & \cite{CohenBlum}       &   14 & \cite{CohenBlum} \\
8 &    9 & \cite{Robbins1979}  &   15 & \cite{CohenBlum}       &   15 & \cite{CohenBlum} \\
9 &   10 & \cite{Robbins1979}  &   17 & \cite{CohenBlum}       &   17 & \cite{CohenBlum} \\
10&   11 & \cite{CohenBlum}    &   18 & \cite{CohenBlum}       &   18 & \cite{CohenBlum} \\
11&   13 & \cite{CohenBlum}    &   19 & \cite{Korf2008}        &   19 & \cite{CohenBlum} \\
12&   14 & \cite{HeydariSudb}  &   21 & \cite{Korf2008}        &   21 & \cite{CohenBlum} \\
13&   15 & \cite{HeydariSudb}  &   22 & Section~\ref{sec:comp} &   22 & \cite{CohenBlum} \\
14&   16 & \cite{Kounoike+2005}&   23 & Section~\ref{sec:comp} &   23 & \cite{CohenBlum} \\
15&   17 & \cite{Kounoike+2005}&   25 & Section~\ref{sec:comp} &   24 & \cite{CohenBlum} \\
16&   18 & \cite{Asai+2006}    &   26 & Section~\ref{sec:comp} &   26 & \cite{CohenBlum} \\
17&   19 & \cite{Asai+2006}    &   28 & Section~\ref{sec:comp} &   28 & \cite{CohenBlum} \\
18&   20 & Section~\ref{sec:comp} &   &                        &   29 & \cite{CohenBlum} \\
19&   22 & Section~\ref{sec:comp} &   &                        &   30 & Section~\ref{sec:comp} \\
20&      &                     &      &                        &   32 & Section~\ref{sec:comp} \\
$n\equiv  3 \pmod 4$ & & & & & $\lfloor\frac{3n+3}{2}\rfloor$ & Corollary~\ref{cor:cbexact} \\
\end{tabular}
\caption{known values of $f(n)$, $g(n)$ and $g(-I_n)$}
\label{table:val}
\end{table}

%\begin{table}[ht]
%\centering
%\begin{tabular}{r rl rl rl}
%%\hline
% & $f(n)$ & & $g(n)$ & & $g(-I_n)$ & \\ 
%\hline
%Lower bound & $15 \l\lfloor \frac{n}{14} \r\rfloor$ & \cite{HeydariSudb} &  
%$\lfloor\frac{3n+3}{2}\rfloor$ & Theorem~\ref{thm:blb} & 
%$\lfloor\frac{3n+3}{2}\rfloor$ & Theorem~\ref{thm:blb} \\
%Upper bound & $\frac{18}{11}n + O(1)$ & \cite{Chitturi+2008} &  
%$2n-2$ & \cite{CohenBlum} & 
%$\frac32 n + O(1)$  & \cite{HeydariSudb} \\
%
%\end{tabular}
%\caption{known bounds on $f(n)$, $g(n)$ and $g(-I_n)$}
%\label{table:bounds}
%\end{table}

%The algorithm of Gates and Papadimitriou \cite{GatesPapad}~ for the burnt version is 2-approximative. 
%A 2-approximative for the unburnt version was given by Fischer and Ginzinger \cite{FischerGinzinger}. 

We present an algorithm that needs on average $7n/4 + O(1)$ flips to sort a stack of $n$ burnt
pancakes and a randomized algorithm for sorting $n$ unburnt pancakes with $17n/12 + O(1)$ flips 
on average. 
%The average number of flips has not been studied yet, therefore the upper bounds 
%for $f$ and $g$ work as upper bounds for the average number of flips as well. 
We also show that any algorithm for the unburnt version requires on average at least $n-O(1)$ flips and 
in the burnt version $n+\Omega(n/\log n)$ flips are needed on average. Section~\ref{sec_concl} 
introduces a conjecture that the average number of flips of the optimal algorithm for sorting burnt 
pancakes is $n+\Theta(n/\log n)$.

\section{Terminology and notation}
The stack obtained by flipping the whole stack $C$ is $\flipped C$. The stack $-C$ 
is obtained from $C$ by changing the orientation of each pancake while keeping 
the order of pancakes.

If two unburnt pancakes of consecutive sizes are located next to each other, they
are \emph{adjacent}. Two burnt pancakes located next to each other are \emph{adjacent} if they form 
a substack of $I_n$ or of $\flipped{I_n}$. Two burnt pancakes located next to each other 
are \emph{anti-adjacent} if they form a substack of $-I_n$ or of $\flipped{-I_n}$.

In both versions a \emph{block} in a stack $C$ is an inclusion-wise maximal substack $S$ of $C$
such that each two pancakes of $S$ on consecutive positions are adjacent.
%In both versions a substack $S$ of a stack $C$ is called a \emph{block}, 
%if each pair of pancakes on consecutive positions is adjacent and 
%extending $S$ on either end will violate this condition. 
A substack $S$ of a stack $C$ with burnt pancakes is 
called a \emph{clan}, if $-S$ is a block in $-C$.
%if each its pair of consecutive pancakes forms an anti-adjacency 
%and extending $S$ on either end will violate this condition. 
Pancake not taking part in a block or a clan is \emph{free}.

If the top $i$ pancakes are flipped, the flip is an \emph{$i$-flip}.

\section{Lower bound in the burnt version}
\label{sec:lb}
%A stack $C$ of pancakes in the burnt version is defined by the permutation of pancakes 
%and their orientations. Pancake $i$ will be represented by $\bsd i$ if its burnt side is down 
%and $\bsu i$ if up. 
%stack $-C$ is obtained from $C$ by changing orientation of each of the pancakes.
\begin{theorem}
\label{thm:blb}
For each $n$ %the number of flips needed to sort the stack $-I_n$ is at least
\[
g(-I_n) \geq \l\lfloor \frac{3(n + 1)}{2}\r\rfloor.
\]
\end{theorem}

\begin{proof}
{\ }\par
The claim is easy to verify for $n \leq 2$, so we can assume $n \geq 3$.

A block (clan) is called a \emph{surface block (clan)} if the topmost pancake is part of it, 
otherwise it is \emph{deep}.

We will assign to each stack $C$ the value $v(C)$:
\[
v(C) := a(C)-a^-(C) - \frac13 (b(C)-b^-(C)) + \frac13 (o(C)-o^-(C)) + l(C)-l^-(C) + \frac13 (ll(C)-ll^-(C)),
\]
where
\begin{align*}
a(C) &:= \text{number of adjacencies} \\
b(C) &:= \text{number of deep blocks} \\
o(C) &:= \l\{ \begin{array}{ll} 
1 & \text{if the pancake on top of the stack is the free $\overline{1}$ or} \\
  & \text{if $1$ is in a block (necessarily with $2$)} \\
0 & \text{otherwise} \\
\end{array}\r. \\
l(C) &:= \l\{ \begin{array}{ll}
1 & \text{if the lowest pancake is $\underline{n}$} \\
0 & \text{otherwise}\\
\end{array}\r. \\
ll(C) &:= \l\{ \begin{array}{ll} 
1 & \text{if the lowest pancake is $\bsd n$ and the second lowest is $\bsd{n-1}$} \\
0 & \text{otherwise}\\
\end{array}\r. \\
a^-(C) &:= a(-C) = \text{number of anti-adjacencies in $C$}\\
b^-(C) &:= b(-C) = \text{number of deep clans in $C$}\\
o^-(C) &:= o(-C) \\
l^-(C) &:= l(-C) \\
ll^-(C) &:= ll(-C). \\
\end{align*}

\begin{lemma}
\label{lem:blb}
If $C$ and $C'$ are stacks of at least two pancakes and $C'$ can be obtained from $C$ by a single flip, then 
\[\Delta v := v(C') - v(C) \leq \frac43.\]
Therefore the minimum number of flips needed to sort a stack $C$ is at least 
\[ \l\lceil \frac34 (v(I_n) - v(C)) \r\rceil .\]
\end{lemma}

\begin{proof}
{\ }\par
First we introduce notation for contributions of each of the functions to $\Delta v$:
\begin{align*}
\Delta a &:= a(C') - a(C) &                   \Delta a^- &:= -(a^-(C') - a^-(C)) \\
\Delta b &:= -\frac13 (b(C') - b(C)) &        \Delta b^- &:= \frac13 (b^-(C') - b^-(C)) \\
\Delta o &:= \frac13 (o(C') - o(C)) &         \Delta o^- &:= -\frac13 (o^-(C') - o^-(C)) \\
\Delta l &:= l(C') - l(C) &                   \Delta l^- &:= -(l^-(C') - l^-(C)) \\
\Delta ll &:= \frac13 (ll(C') - ll(C)) &      \Delta ll^- &:= -\frac13 (ll^-(C') - ll^-(C))
\end{align*}

\begin{observation}
%None of $\Delta a$, $\Delta a^-$, $\Delta l$, $\Delta l^-$ is larger than $1$. 
%None of $\Delta b$, $\Delta b^-$, $\Delta o$, $\Delta o^-$, $\Delta ll$, $\Delta ll^-$ is larger than $\frac13$.
Values of $\Delta a$, $\Delta a^-$, $\Delta l$ and $\Delta l^-$ are among $\{0,1,-1\}$. 
Values of $\Delta b$, $\Delta b^-$, $\Delta o$, $\Delta o^-$, $\Delta ll$ and $\Delta ll^-$ are among 
$\{0, 1/3, -1/3\}$.
\end{observation}
\begin{proof}
The only nontrivial part is $\Delta b \leq 1/3$ and symmetrically $\Delta b^- \leq 1/3$. For contradiction 
suppose $\Delta b > 1/3$, which can only happen when one block was split to two free pancakes and another 
block became surface in a single flip. But the higher of the two pancakes that formed the split block
will end on top of the stack after the flip. Therefore no block became surface. 
To show $\Delta b^- \leq 1/3$ we consider the flip $\phi: -C' \rightarrow -C$, for which
\[
\frac13 \geq \Delta_{\phi}b = -\frac13 (b(-C)-b(-C')) = -\frac13 (b^-(C)-b^-(C')) = \frac13 (b^-(C')-b^-(C)) = \Delta b^-.
\]

%If $\Delta b^- > 1/3$
%we will consider the flip $-C' \rightarrow -C$, for which
%\[
%\frac13 (b(-C')-b(-C)) = -\frac13 (b^-(C)-b^-(C')) = \frac13 (b^-(C')-b^-(C)) = \Delta b^- > \frac13 ,
%\]
%which is impossible, because for the flip $-C' \rightarrow -C$ $(b(-C')-b(-C))/3 \leq 1/3$.

\end{proof}

{\ }\par
The proof of the lemma is based on restricting possible combinations of values of the 
above defined functions.
\begin{itemize}

\item
Both $\Delta l$ and $\Delta l^-$ are positive. This would require the pancake $n$ to be before and after
the flip at the bottom of the stack each time with a different orientation. But this is not possible 
when $n > 1$.

\item
Exactly one of $\Delta l$ and $\Delta l^-$ is positive. The case $\Delta l^- > 0$ can be transformed to 
the case $\Delta l > 0$ by considering the flip $\phi: -C' \rightarrow -C$, for which 
\begin{align*}
 \Delta_{\phi} v &:= v(-C)-v(-C') = -v(C)-(-v(C')) = v(C')-v(C) = \Delta v, \\
 \Delta_{\phi} l &:= l(-C)-l(-C') = l^-(C)-l^-(C') = -(l^-(C')-l^-(C)) = \Delta l^- ,\\
 \Delta_{\phi} l^- &:= l^-(-C)-l^-(-C') = \Delta l.
\end{align*}
The equality $v(-C) = -v(C)$ follows from the definition of $v(C)$.

If the value of $l$ changes, the flip must be an $n$-flip. Therefore $\Delta a = \Delta a^- = 0$.
%If $\Delta l = -1$ , then $\Delta v \leq -1 + 6/3 < 4/3$. 
Because $\Delta l = 1$, the pancake $\bsd n$ has to be at the bottom of the stack after the flip, so 
$\Delta ll^- = 0$. Moreover neither a clan nor the pancake $\bsd 1$ could be on top of the stack 
before the flip so $\Delta b^- \leq 0$ and $\Delta o^- \leq 0$.
Because $\Delta ll = 1/3$ implies a block on top of the stack before the flip and $\Delta o = 1/3$
implies no block on top of the stack after the flip, we obtain
\begin{align*}
\Delta ll = \frac13 ~\&~ \Delta o \leq 0 &\Rightarrow \Delta b \leq 0, \\
\Delta ll \leq 0 ~\&~ \Delta o = \frac13 &\Rightarrow \Delta b \leq 0, \\
\Delta ll = \frac13 ~\&~ \Delta o =\frac13 &\Rightarrow \Delta b \leq -\frac13. 
\end{align*}
In any of the cases $\Delta ll + \Delta o + \Delta b \leq 1/3$ and $\Delta v \leq 4/3$.

From now on, we can assume $\Delta l, \Delta l^- \leq 0$. 

\item
At least one of $\Delta ll$ and $\Delta ll^-$ is positive. If both of them were positive then again 
the pancake $n$ would be at the bottom of the stack before and after the flip, each time with 
a different orientation. Similarly to the previous case, we can choose $\Delta ll^-=0$ and 
$\Delta ll = 1/3$. 
%Similarly to the case with $\Delta l$ and $\Delta l^-$, one of $\Delta ll$ and $\Delta ll^-$ must be zero 
%and we will choose $\Delta ll^-=0$. 
Because $\Delta l \leq 0$, the last flip was an $(n-1)$-flip, 
the pancake at the bottom of the stack is $\bsd n$ and the pancake on top of the stack 
before the flip was $\bsu{(n-1)}$. Therefore $\Delta a = 1$, $\Delta a^- = 0$, 
$\Delta o^- \leq 0$ and $\Delta b^- \leq 0$. 

If pancake $n-1$ was part of a block before the flip, then this block became deep, otherwise 
pancakes $n-1$ and $n$ created a new deep block. Thus $\Delta b \leq 0$.
No block was destroyed and if $\Delta o = 1/3$, then no block became surface and thus 
$\Delta b = -1/3$. All in all $\Delta v \leq 4/3$.

In the remaining cases we have $\Delta l,~\Delta l^-,~\Delta ll,~\Delta ll^- \leq 0$. 

\item
Both $\Delta o$ and $\Delta o^-$ are positive. 
Because $\Delta o^- > 0$ then either 1 was in a clan or on top of the stack with burnt side down before the flip.
If 1 was in a clan, then a single flip would not make it either a part of a block or a free $\bsu 1$ 
on top of the stack and thus $\Delta o$ would not be positive.
Using a similar reasoning for $\Delta o$, we obtain that the flip was a 1-flip, the topmost pancake before 
the flip was $\bsd 1$ and the second pancake from top is different from $2$. 
Thus $\Delta a = \Delta a^- = \Delta b = \Delta b^- = 0$ and $\Delta v \leq 2/3$.

\item
Exactly one of $\Delta o$ and $\Delta o^-$ is positive; without loss of generality it is $\Delta o$. 
This can happen only in two ways.

\begin{itemize}
\item
We did an $i$-flip, the topmost pancake before the flip was $\bsd 2$ and the $(i+1)$-st pancake 
is $\bsu 1$. Then $\Delta a = 1$, $\Delta a^- = 0$, $\Delta b \leq 0$ and $\Delta b^- \leq 0$ and so 
$\Delta v \leq 4/3$.
\item
We did an $i$-flip, the $i$-th pancake before the flip was $\bsd 1$ and neither the $(i-1)$-st nor the 
$(i+1)$-st pancake was $\bsd 2$. Then $\Delta b \leq 0$ and $\Delta a^- \leq 0$. If $\Delta a \leq 0$, 
then $\Delta v \leq 2/3$, otherwise $\Delta b^- \leq 0$ and $\Delta v \leq 4/3$.
\end{itemize}

Now only $\Delta a, \Delta a^-, \Delta b$ and $\Delta b^-$ can be positive.

\item
If $\Delta a = \Delta a^- = 1$, then the flip was either
\[ 
\l(
\begin{array}{c}
	\bsu{i-1} \\ \vdots \\ \bsd{i+1} \\ \bsd{i} \\ \vdots
\end{array}
\r)
\rightarrow
\l(
\begin{array}{c}
	\bsu{i+1} \\ \vdots \\ \bsd{i-1} \\ \bsd{i} \\ \vdots
\end{array}
\r)
\text{\qquad, or \qquad}
\l(
\begin{array}{c}
	\bsd{i+1} \\ \vdots \\ \bsu{i-1} \\ \bsu{i} \\ \vdots
\end{array}
\r)
\rightarrow
\l(
\begin{array}{c}
	\bsd{i-1} \\ \vdots \\ \bsu{i+1} \\ \bsu{i} \\ \vdots
\end{array}
\r).
\]

In both cases the topmost pancake before the flip was not part of a clan and the topmost pancake after the 
flip is not part of a block, so the number of deep blocks increased and the number of deep clans
decreased and $\Delta v \leq 4/3$.

\item
Exactly one of $\Delta a$ and $\Delta a^-$ is positive; without loss of generality $\Delta a = 1$, $\Delta a^- \leq 0$.
Neither a new clan was created, nor became deep, so $\Delta b^- \leq 0$ and $\Delta v \leq 4/3$.

\item
None of $\Delta a$ and $\Delta a^-$ is positive, so $\Delta v \leq 2/3$.

\end{itemize}
\end{proof}

It is easy to compute that $v(I_n)=n+2/3$ and $v(-I_n)=-n-2/3$ and thus 
the number of flips needed to transform $-I_n$ to $I_n$ is at least 
\[
\l\lceil \frac34 \l(v(I_n) - v(-I_n)\r) \r\rceil =
\l\lceil \frac34 \l(2n+\frac43\r)\r\rceil = \l\lceil \frac32n + 1\r\rceil =
\l\lfloor \frac{3(n + 1)}{2}\r\rfloor
.
\]

\end{proof}

\begin{corollary}
\label{cor:cbexact}
For all integers $n\geq 15$ with $n \equiv 3 \pmod 4$,
\[
g(-I_n) = \l\lfloor \frac{3(n + 1)}{2}\r\rfloor.
\]
\end{corollary}
\begin{proof}
The lower bound comes from Theorem~\ref{thm:blb}. For all $n\geq 23$ with $n \equiv 3 \pmod 4$, 
the upper bound was proved by Heydari and Sudborough~\cite{HeydariSudb}. The exact value 
for $n=15$ was computed by Cohen and Blum~\cite{CohenBlum} and the exact value for $n=19$ 
is computed in Section~\ref{sec:comp}.
\end{proof}

\section{Algorithm for the burnt version}
\label{sec:avb}

In this section we will design an algorithm that sorts burnt pancakes with small average number of flips.

First we will show a lower bound on the average number of flips of any algorithm that sorts a stack of
$n$ burnt pancakes.

\begin{theorem}
\label{thm:avgblb}
Let $av_{opt}(n)$ be the average number of flips of the optimal algorithm for sorting a stack of $n$ 
burnt pancakes. For any $n\geq 16$
\[ av_{opt}(n) \geq n + \frac{n}{16\log_2 n} - \frac32. \]
\end{theorem}

\begin{proof}
We will first count the expected number of adjacencies in a stack of $n$ burnt pancakes. 
A stack has $n-1$ pairs of pancakes on consecutive positions.
For each such pair of pancakes, there are $4 n (n-1)$ equally probable combinations of 
their values and orientations and the pancakes form an adjacency in exactly $2(n-1)$ of them.
From the linearity of expectation
\[\mathbb E[adj] = (n-1)\frac{1}{2n} = \frac12 \frac{n-1}{n}. \] 
Therefore at least half of the stacks have no adjacency.

\begin{itemize}
\item
First we take a half of the stacks such, that it contains all the stacks which have some adjacency.
The stacks of this half have less than 1 adjacency on average.
Each flip creates at most one adjacency, therefore when we want to obtain the stack $I_n$ with $n-1$ adjacencies,
we need at least $n-2$ flips on average.

\item
The other half contains $n! \cdot 2^{n-1}$ stacks each with no adjacency, 
thus requiring at least $n-1$ flips. For each stack we take one of the shortest sequences of flips 
that create the stack from $I_n$ and call it the \emph{creating sequence} of the stack.
Note that creating sequences of two different stacks are different.
We will now count the number of different creating sequences of length at most 
$n-1+n/(4\log_2 n)$, which will give an upper bound on the 
number of stacks with no adjacency that can be sorted in $n-1+n/(4\log_2 n)$ flips. 
Shorter creating sequences will be followed by several 0-flips, therefore we will consider 
$n+1$ possible flips. A \emph{split-flip} is a flip in a creating sequence that decreases the number 
of adjacencies to a value smaller than the lowest value obtained before the flip. Therefore 
there are exactly $n-1$ split-flips in each of our creating sequences. In a creating sequence, 
the $i$-th split-flip removes one of $n-i$ existing adjacencies and therefore there are 
$n-i$ possibilities how to make the $i$-th split-flip. The number of different creating 
sequences of the above given length is at most 
\begin{align*}
& \binom{n-1+\frac{n}{4\log_2 n}}{\frac{n}{4\log_2 n}}\cdot(n-1)! \cdot (n+1)^{n/(4\log_2 n)} \\
& \leq \l({n-1+\frac{n}{4\log_2 n}}\r)^{n/(4\log_2 n)} \cdot (n-1)! \cdot (2n)^{n/(4\log_2 n)} \\
& \leq (n-1)! \cdot (2n)^{n/(4\log_2 n)} \cdot (2n)^{n/(4\log_2 n)} \\
& \leq (n-1)! \cdot \l(n^{5/4}\r)^{2n/(4\log_2 n)} \\
& \leq (n-1)! \cdot 2^{5n/8} \\
& < \frac14 n! \cdot 2^n.
\end{align*}
Thus at least half of the stacks with no adjacency need more than $n-1+n/(4\log_2 n)$ flips while 
the rest needs at least $n-1$ flips. Therefore in this case the average number of flips is at least
\[
n-1+\frac{n}{8\log_2 n}.
\]
\end{itemize}

The overall average number of flips is then
\[
av_{opt}(n) \geq n - \frac32 + \frac{n}{16\log_2 n}.
\]

\end{proof}

\begin{theorem}
\label{thm:balgo}
There exists an algorithm that sorts a stack of $n$ burnt pancakes with the average number of flips at most
\[
\frac74 n + 5.
\]
\end{theorem}

\begin{proof}
Let $\mathbb C_n$ denote the set of all stacks of $n$ burnt pancakes, $h(C)$ 
will be the number of flips used by the algorithm to sort the stack $C$ and let
\begin{align*}
H(n) &:= \sum_{C \in \mathbb C_n} h(C), \\
av(n) &:= \frac{H(n)}{|\mathbb C_n|} = \frac{H(n)}{2n|\mathbb C_{n-1}|}.
\end{align*}

The algorithm will never break previously created adjacencies. This allows us to consider the adjacent
pancakes as a single burnt pancake. In each iteration of the algorithm one adjacency is created, 
the two adjacent pancakes are contracted and the size of the stack
decreases by one. We stop when the number of pancakes is two and the algorithm can transform the stack 
to the stack $(\bsd 1)$ in at most four flips.

However for the simplicity of the discussion, we will not do such a 
contraction for adjacencies already existing in the input stack (as can be seen in the proof 
of Theorem~\ref{thm:avgblb}, there are very few such adjacencies, so the benefit would be negligible).

One more simplification is used. 
Before each iteration, the algorithm looks at the topmost pancake and cyclically renumbers the pancakes 
so as to have the topmost pancake numbered $2$ --- pancake number $j$ will become $j+s+kn$, 
where $s=(2-\pi(1))$ and $k$ is an integer chosen so as to have the result inside 
the interval $\{1,\dots,n\}$. Let $\mathbb C^2_{n}$ be the set of stacks with $n$ 
burnt pancakes and the pancake number 2 on top. 
When we end up with the stack $(\bsd 1)$, we in fact have

\[ 
\l(
\begin{array}{c}
	\bsd{i} \\ \bsd{i+1} \\ \vdots \\ \bsd{n} \\ \bsd{1}\\ \bsd{2}\\ \vdots \\ \bsd{i-1}
\end{array}
\r),
\]
for some $i \in \{1,2, \dots n \}$. This stack needs at most four more flips to become $I_n$. 
Therefore $av(2) \leq 8$. We will do four flips at the end even if they are not necessary. Then 
the number of flips will not be changed by a cyclic renumbering of pancakes and 
$H(n) = n \cdot \sum_{C \in \mathbb C^2_n} h(C)$.

\begin{itemize}
\item 
If the stack from $\mathbb C^2_{n}$ can be flipped so that the topmost pancake will form an adjacency, we will do it:
\[ 
\l(
\begin{array}{c}
	\bsd{2} \\ X \\ \bsu{1} \\ Y 
\end{array}
\r)
\rightarrow
\l(
\begin{array}{c}
	\flipped X \\ \bsu{2} \\ \bsu{1} \\ Y 
\end{array}
\r)
\Leftrightarrow
\l(
\begin{array}{c}
	X' \\ \bsu{1} \\ Y '
\end{array}
\r)
\in \mathbb C_{n-1},
\]

or

\[
\l(
\begin{array}{c}
	\bsu{2} \\ X \\ \bsd{3} \\ Y 
\end{array}
\r)
\rightarrow
\l(
\begin{array}{c}
	\flipped X \\ \bsd{2} \\ \bsd{3} \\ Y 
\end{array}
\r)
\Leftrightarrow
\l(
\begin{array}{c}
	X' \\ \bsd{2} \\ Y '
\end{array}
\r)
\Leftrightarrow
\l(
\begin{array}{c}
	X'' \\ \bsd{1} \\ Y''
\end{array}
\r)
\in \mathbb C_{n-1}.
\]

Each stack from $\mathbb C_{n-1}$ appears as a result of the above described process for exactly 
one stack from $\mathbb C^2_{n}$.

\item 
%If the stack from $\mathbb C^2_{n}$ can not be flipped so that the topmost pancake will form an adjacency, 
%we will look at both pancakes $1$ and $3$ and analyze all possible cases. Note that when $2$ has its 
%burnt side up, then $3$ has its burnt side up and similarly $\bsd 2$ implies $\bsd 1$.
If no adjacency can be created in a single flip, we will look at both pancakes $1$ and $3$ and analyze 
all possible cases. Note that this time when $2$ has its burnt side up, then $3$ has its burnt side up 
and similarly $\bsd 2$ implies $\bsd 1$.

\begin{enumerate}
\item
\[ 
\l(
\begin{array}{c}
	\bsd{2} \\ X \\ \bsd{1} \\ Y \\ \bsd{3} \\ Z
\end{array}
\r)
\rightarrow
\l(
\begin{array}{c}
	\bsu{2} \\ X \\ \bsd{1} \\ Y \\ \bsd{3} \\ Z
\end{array}
\r)
\rightarrow
\l(
\begin{array}{c}
	\flipped Y \\ \bsu{1} \\ \flipped X \\ \bsd{2} \\ \bsd{3} \\ Z
\end{array}
\r)
\Leftrightarrow
\l(
\begin{array}{c}
	Y' \\ \bsu{1} \\ X' \\ \bsd{2} \\ Z'
\end{array}
\r)
\in \mathbb C_{n-1}
\]

\item
\[ 
\l(
\begin{array}{c}
	\bsd{2} \\ X \\ \bsd{3} \\ Y \\ \bsd{1} \\ Z
\end{array}
\r)
\rightarrow
\l(
\begin{array}{c}
	\bsu{2} \\ X \\ \bsd{3} \\ Y \\ \bsd{1} \\ Z
\end{array}
\r)
\rightarrow
\l(
\begin{array}{c}
	\flipped X \\ \bsd{2}\\ \bsd{3} \\ Y  \\ \bsd{1} \\ Z
\end{array}
\r)
\Leftrightarrow
\l(
\begin{array}{c}
	X' \\ \bsd{2} \\ Y' \\ \bsd{1} \\ Z'
\end{array}
\r)
\in \mathbb C_{n-1}
\]

\item
\[ 
\l(
\begin{array}{c}
	\bsd{2} \\ X \\ \bsd{1} \\ Y \\ \bsu{3} \\ Z
\end{array}
\r)
\rightarrow
\l(
\begin{array}{c}
	\bsd{3} \\ \flipped Y \\ \bsu{1} \\ \flipped X \\ \bsu{2} \\ Z
\end{array}
\r)
\rightarrow
\l(
\begin{array}{c}
	X \\ \bsd{1}\\ Y \\ \bsu{3} \\ \bsu{2} \\ Z
\end{array}
\r)
\Leftrightarrow
\l(
\begin{array}{c}
	X' \\ \bsd{1} \\ Y' \\ \bsu{2} \\ Z'
\end{array}
\r)
\in \mathbb C_{n-1}
\]

\item
\[ 
\l(
\begin{array}{c}
	\bsd{2} \\ X \\ \bsu{3} \\ Y \\ \bsd{1} \\ Z
\end{array}
\r)
\rightarrow
\l(
\begin{array}{c}
	\bsd{3} \\ \flipped X \\ \bsu{2} \\ Y \\ \bsd{1} \\ Z
\end{array}
\r)
\rightarrow
\l(
\begin{array}{c}
	X \\ \bsu{3}\\ \bsu{2} \\ Y \\ \bsd{1} \\ Z
\end{array}
\r)
\Leftrightarrow
\l(
\begin{array}{c}
	X' \\ \bsu{2} \\ Y' \\ \bsd{1} \\ Z'
\end{array}
\r)
\in \mathbb C_{n-1}
\]

\item
\[ 
\l(
\begin{array}{c}
	\bsu{2} \\ X \\ \bsu{3} \\ Y \\ \bsd{1} \\ Z
\end{array}
\r)
\rightarrow
\l(
\begin{array}{c}
	 \bsu{1} \\ \flipped Y \\ \bsd{3} \\ \flipped X \\ \bsd{2} \\  Z
\end{array}
\r)
\rightarrow
\l(
\begin{array}{c}
	X \\ \bsu{3} \\ Y \\ \bsd{1} \\ \bsd{2} \\  Z
\end{array}
\r)
\Leftrightarrow
\l(
\begin{array}{c}
	X' \\ \bsu{2} \\ Y' \\ \bsd{1} \\ Z'
\end{array}
\r)
\rightarrow
\l(
\begin{array}{c}
	\flipped Z' \\ \bsu{1} \\ \flipped Y' \\ \bsd{2} \\ \flipped X'
\end{array}
\r)
\rightarrow
\l(
\begin{array}{c}
	Y' \\ \bsd{1} \\ Z' \\ \bsd{2} \\ \flipped X'
\end{array}
\r)
\in \mathbb C_{n-1}
\]

\item
\[ 
\l(
\begin{array}{c}
	\bsu{2} \\ X \\ \bsd{1} \\ Y \\ \bsu{3} \\ Z
\end{array}
\r)
\rightarrow
\l(
\begin{array}{c}
	 \bsu{1} \\ \flipped X \\ \bsd{2} \\  Y \\ \bsu{3} \\ Z
\end{array}
\r)
\rightarrow
\l(
\begin{array}{c}
	 X \\ \bsd{1} \\ \bsd{2} \\  Y \\ \bsu{3} \\ Z
\end{array}
\r)
\Leftrightarrow
\l(
\begin{array}{c}
	X' \\ \bsd{1} \\ Y' \\ \bsu{2} \\ Z'
\end{array}
\r)
\rightarrow
\l(
\begin{array}{c}
	 \flipped Z' \\ \bsd{2} \\ \flipped Y' \\ \bsu{1} \\ \flipped X'
\end{array}
\r)
\rightarrow
\l(
\begin{array}{c}
	Y' \\ \bsu{2} \\ Z' \\ \bsu{1} \\ \flipped X'
\end{array}
\r)
\in \mathbb C_{n-1}
\]

\item
\[ 
\l(
\begin{array}{c}
	\bsu{2} \\ X \\ \bsu{3} \\ Y \\ \bsu{1} \\ Z
\end{array}
\r)
\rightarrow
\l(
\begin{array}{c}
	\bsd{2} \\ X \\ \bsu{3} \\ Y \\ \bsu{1} \\ Z
\end{array}
\r)
\rightarrow
\l(
\begin{array}{c}
	\flipped Y \\ \bsd{3} \\ \flipped X \\ \bsu{2} \\ \bsu{1} \\ Z
\end{array}
\r)
\Leftrightarrow
\l(
\begin{array}{c}
	Y' \\ \bsd{2} \\ X' \\ \bsu{1} \\ Z'
\end{array}
\r)
\in \mathbb C_{n-1}
\]

\item
\[ 
\l(
\begin{array}{c}
	\bsu{2} \\ X \\ \bsu{1} \\ Y \\ \bsu{3} \\ Z
\end{array}
\r)
\rightarrow
\l(
\begin{array}{c}
	\bsd{2} \\ X \\ \bsu{1} \\ Y \\ \bsu{3} \\ Z
\end{array}
\r)
\rightarrow
\l(
\begin{array}{c}
	\flipped X \\ \bsu{2} \\ \bsu{1} \\ Y \\ \bsu{3} \\ Z
\end{array}
\r)
\Leftrightarrow
\l(
\begin{array}{c}
	X' \\ \bsu{1} \\ Y' \\ \bsu{2} \\ Z'
\end{array}
\r)
\in \mathbb C_{n-1}
\]

\end{enumerate}

Again each stack from $\mathbb C_{n-1}$ appears as a result of the process for exactly one stack 
from $\mathbb C^2_{n}$, but we needed two additional flips in two of the cases to ensure this.
We did four flips in a quarter of the cases and two flips in all other cases. Each
case has the same probability and hence the average number of flips is $5/2$.

\end{itemize}

All in all
\begin{align*}
H(n) &= n \cdot \l(\sum_{C \in \mathbb C_{n-1}} (h(C)+1) + \sum_{C \in \mathbb C_{n-1}} \l(h(C)+\frac52 \r) \r)
 = 2n H(n-1) + \frac72 n |C_{n-1}|, \\
av(n) &= \frac{2nH(n-1) + \frac72 n |C_{n-1}|}{2n|C_{n-1}|} = av(n-1) + \frac74 = av(2) + \frac74(n-2) \leq \frac74 n + 5.
\end{align*}
\end{proof}

\section{Randomized algorithm for the unburnt version}
\label{sec:avu}

\begin{observation}
Let $av'_{opt}(n,0)$ be the average number of flips of the optimal algorithm for sorting a stack of $n$ 
unburnt pancakes. For any positive $n$
\[ av'_{opt}(n,0) \geq n-2. \]
\end{observation}
\begin{proof}
We will now count the expected number of adjacencies in a stack of $n$ pancakes. 
For the purpose of this proof we will consider the pancake number $n$ at the bottom of the stack as 
an additional adjacency; this has probability $1/n$.
Pancakes on consecutive positions form an adjacency if their values differ by $1$; the probability of this is $2/n$.
Therefore the expected number of adjacencies is
\[\mathbb E[adj] = \frac{1}{n} + (n-1)\frac{2}{n} < 2. \]

Each flip creates at most one adjacency, therefore when we want to obtain the stack $I_n$ with $n$ adjacencies,
the average number of flips is at least $n-2$.
\end{proof}

\begin{theorem}
\label{thm:ualgo}
There exists a randomized algorithm that sorts a stack of $n$ unburnt pancakes with the average number of
flips at most
\[
\frac{17}{12}n + 9,
\]
where the average is taken both over the stacks and the random bits.
\end{theorem}

\begin{proof}
If two pancakes become adjacent, we contract them to a single burnt pancake; its burnt side 
will be the one where the pancake with higher number was.
Therefore in the course of the algorithm, some of
the pancakes will be burnt and some unburnt. For this reason we say that two pancakes are \emph{adjacent}
if the unburnt ones of them can be oriented so that the two resulting pancakes satisfy the definition 
of adjacency for burnt pancakes.

Let $\mathbb U_{n,b}$ denote the set of all stacks of $n$ pancakes $b$ of which are burnt and
let $\mathbb U^{2}_{n,b}$ be the stacks from $\mathbb U_{n,b}$ with the pancake number 2 on top.
Let $k(C)$ be the number of flips needed by the algorithm to sort the stack $C$ and let
\begin{align*}
K(n,b) &:= \sum_{C \in \mathbb U_{n,b}} k(C), \\
av'(n,b) &:= \frac{K(n,b)}{|\mathbb U_{n,b}|}.
\end{align*}

When there are only two pancakes left, we can sort the stack in at most 4 flips.
Similarly to the burnt version, we will sometimes cyclically renumber the pancakes. After 
renumbering them back at the end, we will do 4 flips to get the sorted stack.
Therefore $av'(1,0) = av'(1,1) = 4$, $av'(2,b) \leq 8$ for any $b \in \{0,1,2\}$
and $K(n,b) = n \cdot \sum_{C \in \mathbb U^{2}_{n,b}} k(C)$.

The algorithm first cyclically renumbers the pancakes so as to have the topmost 
pancake numbered 2 thus obtaining a stack from $\mathbb U^{2}_{n,b}$. Then we look at the 
topmost pancake. If it is unburnt, we uniformly at random select whether to look at 1 or 3; 
if it is burnt and the burnt side is down, we look at 1 and in the case when the burnt side 
is up, we look at 3.

Notice that we could also look at both pancakes 1 and 3. But if we joined only two of the 
pancakes 1, 2 and 3 we would have to count the average number of flips for each combination 
not only of the number of pancakes and the number of burnt pancakes, but also of the number 
of pairs of pancakes of consecutive sizes exactly one of which is burnt. This would make 
the calculations too complicated. We could also join all three of them, but this would 
lead to a worse result.

%Notice that we could also look at both pancakes 1 and 3, but then we would have to 
%count the average number of flips for each combination not only of the number of pancakes 
%and the number of burnt pancakes, but also of the number of pairs of pancakes of consecutive 
%sizes exactly one of which is burnt, which would make the calculations too complicated.
%We could also join all three of them, but this would lead to a worse result.
%Notice that we could also look at both pancakes 1 and 3, but then we would have to take into 
%consideration not only the number of pancakes and the number of burnt pancakes, but also 
%the number of pairs of pancakes such that their values differ by 1 and exactly one of them is burnt.

\begin{enumerate}[I.]
\item
\label{case1}
Both the pancakes we looked at are unburnt. The set of such stacks is 
$\mathbb U^{2,\ref{case1}}_{n,b}$. Note that stacks with pancake 2 unburnt and exactly 
one of pancakes 1 and 3 unburnt belong to this set from $50\%$ --- with $50\%$ probability, 
we choose to look at the unburnt pancake. 
Let $av'_{\ref{case1}}(n,b)$ be the weighted average number of flips used by 
the algorithm to sort a stack from $\mathbb U^{2,\ref{case1}}_{n,b}$, where the weight 
is the ratio with which the stack belongs to $\mathbb U^{2,\ref{case1}}_{n,b}$.

%\[
%av'_{\ref{case1}}(n,b) := \frac{\sum_{C \in \mathbb U_{n,b}} k(C) \prob(C \in \mathbb U^{2,\ref{case1}}_{n,b})}
%{|\mathbb U^{2,\ref{case1}}_{n,b}|}.
%\]

\[
\l(
\begin{array}{c}
	2 \\ X \\ 1 \\ Y
\end{array}
\r)
\rightarrow
\l(
\begin{array}{c}
	\flipped X \\ 2 \\ 1 \\ Y
\end{array}
\r)
\Leftrightarrow
\l(
\begin{array}{c}
	X' \\ \bsu{1} \\ Y'
\end{array}
\r)
\in \mathbb U_{n-1,b+1}
\]

\[
\l(
\begin{array}{c}
	2 \\ X \\ 3 \\ Y
\end{array}
\r)
\rightarrow
\l(
\begin{array}{c}
	\flipped X \\ 2 \\ 3 \\ Y
\end{array}
\r)
\Leftrightarrow
\l(
\begin{array}{c}
	X' \\ \bsd{2} \\ Y'
\end{array}
\r)
\Leftrightarrow
\l(
\begin{array}{c}
	X'' \\ \bsd{1} \\ Y''
\end{array}
\r)
\in \mathbb U_{n-1,b+1}
\]

For each stack from $\mathbb U_{n-1,b+1}$ there are exactly $b+1$ its cyclic renumberings each appearing
as a result with a $50\%$ probability. Thus we can compute the average number of 
flips in this case:
\[
av'_{\ref{case1}}(n,b) = av'(n-1,b+1)+1 .
\]

\item
\label{case2}
The topmost pancake is unburnt, while the other pancake we looked at is burnt.
\[
\l(
\begin{array}{c}
	2 \\ X \\ \bsu{1} \\ Y
\end{array}
\r)
\rightarrow
\l(
\begin{array}{c}
	\flipped X \\ 2 \\ \bsu{1} \\ Y
\end{array}
\r)
\Leftrightarrow
\l(
\begin{array}{c}
	X' \\ \bsu{1} \\ Y'
\end{array}
\r)
\in \mathbb U_{n-1,b}
\]

\[
\l(
\begin{array}{c}
	2 \\ X \\ \bsd{1} \\ Y
\end{array}
\r)
\rightarrow
\l(
\begin{array}{c}
	\bsu{1} \\ \flipped X \\ 2 \\ Y
\end{array}
\r)
\rightarrow
\l(
\begin{array}{c}
	X \\ \bsd{1} \\ 2 \\ Y
\end{array}
\r)
\Leftrightarrow
\l(
\begin{array}{c}
	X' \\ \bsd{1} \\ Y'
\end{array}
\r)
\in \mathbb U_{n-1,b}
\]

The case when we looked at pancake $3$ is similar, so we can conclude that 
\[
av'_{\ref{case2}}(n,b) = av'(n-1,b) + \frac32 .
\]

\item
\label{case3}
The topmost pancake is burnt, while the other one we looked at is unburnt.
\[
\l(
\begin{array}{c}
	\bsu{2} \\ X \\ 3 \\ Y
\end{array}
\r)
\rightarrow
\l(
\begin{array}{c}
	\flipped X \\ \bsd{2} \\ 3 \\ Y
\end{array}
\r)
\Leftrightarrow
\l(
\begin{array}{c}
	X' \\ \bsd{2} \\ Y'
\end{array}
\r)
\Leftrightarrow
\l(
\begin{array}{c}
	X'' \\ \bsd{1} \\ Y''
\end{array}
\r)
\in \mathbb U_{n-1,b}
\]

\[
\l(
\begin{array}{c}
	\bsd{2} \\ X \\ 1 \\ Y
\end{array}
\r)
\rightarrow
\l(
\begin{array}{c}
	\flipped X \\ \bsu{2} \\ 1 \\ Y
\end{array}
\r)
\Leftrightarrow
\l(
\begin{array}{c}
	X'' \\ \bsu{1} \\ Y''
\end{array}
\r)
\in \mathbb U_{n-1,b}
\]

Each stack from $\mathbb U_{n-1,b}$ appears as a result exactly once for $b$ its cyclic renumberings. 
Therefore
\[
av'_{\ref{case3}}(n,b) = av'(n-1,b) + 1 .
\]

\item
\label{case4}
Both the pancakes we looked at are burnt.
In half of the cases the two pancakes can be joined in a single flip:

\[
\l(
\begin{array}{c}
	\bsu{2} \\ X \\ \bsd{3} \\ Y
\end{array}
\r)
\rightarrow
\l(
\begin{array}{c}
	\flipped X \\ \bsd{2} \\ \bsd{3} \\ Y
\end{array}
\r)
\Leftrightarrow
\l(
\begin{array}{c}
	X' \\ \bsd{2} \\ Y'
\end{array}
\r)
\Leftrightarrow
\l(
\begin{array}{c}
	X'' \\ \bsd{1} \\ Y''
\end{array}
\r)
\in \mathbb U_{n-1,b-1}
\]

\[
\l(
\begin{array}{c}
	\bsd{2} \\ X \\ \bsu{1} \\ Y
\end{array}
\r)
\rightarrow
\l(
\begin{array}{c}
	\flipped X \\ \bsu{2} \\ \bsu{1} \\ Y
\end{array}
\r)
\Leftrightarrow
\l(
\begin{array}{c}
	X'' \\ \bsu{1} \\ Y''
\end{array}
\r)
\in \mathbb U_{n-1,b-1}
\]

Otherwise we need three flips to join the two pancakes:

\[
\l(
\begin{array}{c}
	\bsu{2} \\ X \\ \bsu{3} \\ Y
\end{array}
\r)
\rightarrow
\l(
\begin{array}{c}
	\bsd{2} \\ X \\ \bsu{3} \\ Y
\end{array}
\r)
\rightarrow
\l(
\begin{array}{c}
	\bsd{3} \\ \flipped X \\ \bsu{2} \\ Y
\end{array}
\r)
\rightarrow
\l(
\begin{array}{c}
	X \\ \bsu{3} \\ \bsu{2} \\ Y
\end{array}
\r)
\Leftrightarrow
\l(
\begin{array}{c}
	X' \\ \bsu{2} \\ Y'
\end{array}
\r)
\Leftrightarrow
\l(
\begin{array}{c}
	X'' \\ \bsu{1} \\ Y''
\end{array}
\r)
\in \mathbb U_{n-1,b-1}
\]

\[
\l(
\begin{array}{c}
	\bsd{2} \\ X \\ \bsd{1} \\ Y
\end{array}
\r)
\rightarrow
\l(
\begin{array}{c}
	\bsu{2} \\ X \\ \bsd{1} \\ Y
\end{array}
\r)
\rightarrow
\l(
\begin{array}{c}
	\bsu{1} \\ \flipped X \\ \bsd{2} \\ Y
\end{array}
\r)
\rightarrow
\l(
\begin{array}{c}
	X \\ \bsd{1} \\ \bsd{2} \\ Y
\end{array}
\r)
\Leftrightarrow
\l(
\begin{array}{c}
	X'' \\ \bsd{1} \\ Y''
\end{array}
\r)
\in \mathbb U_{n-1,b-1}
\]

Altogether
\[
av'_{\ref{case4}}(n,b) = av'(n-1,b-1) + 2.
\]

\end{enumerate}

After summing up all the above average numbers of flips 
multiplied by their probabilities, we obtain:
\begin{itemize}
\item
For $1\leq b < n$
\begin{align*}
av'(n,b) =& 
\frac{(n-b)(n-b-1)}{n(n-1)} av'_{\ref{case1}}(n,b) + 
\frac{(n-b)b} {n(n-1)} \l( av'_{\ref{case2}}(n,b) + av'_{\ref{case3}}(n,b) \r) + \\
&+ \frac {b(b-1)}{n(n-1)}av'_{\ref{case4}}(n,b) = \\
=& \frac{(n-b)(n-b-1)}{n(n-1)} (1 + av'(n-1,b+1)) + \\
&+ 2\frac{(n-b)b} {n(n-1)} \l( \frac54 + av'(n-1,b) \r) + \frac {b(b-1)}{n(n-1)} \l( 2+av'(n-1,b-1)  \r)
.
\end{align*}
\item
For $b=0$
\[
av'(n,0) = \frac{n(n-1)}{n(n-1)} av'_{\ref{case1}}(n,0) = 1 + av'(n-1,1).
\]
\item
For $b=n$
\[
av'(n,n) = \frac {n(n-1)}{n(n-1)}av'_{\ref{case4}}(n,n)  = 2 + av'(n-1,n-1).
\]
\end{itemize}

Instead of solving these recurrent formulas, we will use them to bound $av'(n,b)$ from above by 
the following function:

\[
av^+(n,b) := \frac{17}{12}n + \frac{7}{12}b - \frac16 \frac{(n-b+1)b}{n} + 9.
\]

\begin{lemma}
For any nonnegative $n$ and $b$, such that $b$ is not greater than $n$
\[av^+(n,b)\geq av'(n,b).\]
\end{lemma}

\begin{proof}
We will use induction on the number of pancakes.

\begin{itemize}
\item
For $n=1$ we have $av'(1,b)=4$ and it is easy to verify that the lemma holds.
%Because in each of the cases \ref{case1} - \ref{case4} no more than two flips were needed on average to 
%decrease the number of pancakes by one, $2n-4$ flips give a stack with two pancakes. This gives 
%a bound $av'(n,b) \leq 2n+4$. If $n$ is at most five, then
%$av'(n,b) \leq n + 9 \leq av^+(n,b)$.
\item
If $b=0$, then the induction hypothesis gives
\begin{align*}
av'(n,0) &= 1 + av'(n-1,1) \leq 1 + av^+(n-1,1) = \\
&= 1 + \frac{17}{12}(n-1) + \frac{7}{12} - \frac16 \frac{n-1}{n-1} + 9 = \frac{17}{12}n + 9 = av^+(n,0).
\end{align*}

\item
For $b=n$ we get
\begin{align*}
av'(n,n) &= 2 + av'(n-1,n-1) \leq 2 + av^+(n-1,n-1) = \\
&= 2 + \frac{17}{12}(n-1) + \frac{7}{12}(n-1) - \frac16 + 9 =
 \frac{17}{12}n + \frac{7}{12}n - \frac16 + 9 = av^+(n,n).
\end{align*}

\item
In the case $1 \leq b < n$
\begin{align*}
n (n-1)&(av^+(n,b) - av'(n,b)) \\
\geq &n(n-1)av^+(n,b) - (n-b)(n-b-1) (1 + av^+(n-1,b+1)) \\
&- 2(n-b)b \l( \frac54 + av^+(n-1,b) \r) - b(b-1) \l( 2+av^+(n-1,b-1) \r) \\
%= & -\frac16 (n-1)(n-b+1)b \\
%&+ (n-b)(n-b-1) \l(\frac{17}{12} - \frac{7}{12} + \frac16 \frac{(n-b-1)(b+1)}{n-1} - 1 )\r) \\
%&+ 2(n-b)b \l( \frac{17}{12} + \frac16 \frac{(n-b) b}{n-1} - \frac54 \r) \\
%&+ b(b-1) \l(\frac{17}{12} + \frac{7}{12} + \frac16 \frac{(n-b+1)(b-1)}{n-1} - 2 \r) \\
= & \frac{b}{n-1}\l(\frac13 n - \frac13 b\r) > 0.
\end{align*}
\end{itemize}

\end{proof}

Therefore $av^+(n,b) \geq av'(n,b)$ and thus 
\[
av'(n,0) \leq av^+(n,0) = \frac{17}{12}n + 9.
\]

\end{proof}

\section{Computational results}
\label{sec:comp}
Computer search found the following sequence of 30 flips that sorts the stack $-I_{19}$: 
(19, 14, 7, 4, 10, 18, 6, 4, 10, 19, 14, 4, 9, 11, 8, 18, 8, 11, 9, 4, 14, 19, 10, 4, 6, 18, 10, 4, 7, 14).
Thus, using Theorem~\ref{thm:blb}, $g(-I_{19}) = 30$.

We also computed $g(-I_{20}) = 32$: From~\cite[Theorem 7]{CohenBlum}: $g(-I_{20}) \leq g(-I_{19}) + 2 = 32$.
From Theorem~\ref{thm:blb}: $g(-I_{20}) \geq 31$ and from Lemma~\ref{lem:blb} follows that if
$g(-I_{20}) = 31$, then each flip of the optimal sorting sequence increases the value of 
the function $v$ by $4/3$. But computer search revealed that starting at $-I_{20}$ we can make a sequence 
of only at most 29 such flips.

The values $f(18)=20$ and $f(19)=22$ were computed by the method of Kounoike et al.~\cite{Kounoike+2005} 
and Asai et al.~\cite{Asai+2006}. It is an improvement of the method of Heydari and Sudborough~\cite{HeydariSudb}.
Let $\mathbb U_{n}^{m}$ be the set of stacks of $n$ unburnt pancakes requiring $m$ flips to sort. 
%The method 
%is based on determining $\mathbb U_{n}^{m}$ from the sets $\mathbb U_{n-1}^{m'}$ with $m' \in \{m-2, m-1,\dots\}$.
For every stack $U \in \mathbb U_{n}^{m}$, $2$ flips always suffice to move the largest pancake to the bottom of 
the stack, obtaining stack $U'$. Since then, it never helps to move the largest pancake. Therefore $U'$ requires 
exactly the same number of flips as $U''$ obtained from $U'$ by removing the largest pancake and thus 
$U''$ requires at least $m-2$ flips.

To determine $\mathbb U_{n}^{i}$ for all $i \in \{m, m+1,\dots, f(n)\}$, it is thus enough to consider 
the set $\cup_{m'=m-2}^{f(n-1)}\mathbb U_{n-1}^{m'}$. 
In each stack from this set, we try adding the pancake number $n$ to the bottom, flipping the whole stack
and trying every possible flip. The candidate set composed of the resulting and the intermediate stacks 
contains all the stacks from $\cup_{i=m}^{f(n)}\mathbb U_{n}^{i}$. 
Now it remains to determine the value of $f(U)$ for each stack $U$ in the candidate set.
As in~\cite{Kounoike+2005} and~\cite{Asai+2006}, this is done using the A* search. 

During the A* search, we need to compute a lower bound on the number of flips needed to sort
a stack. It is counted differently then in~\cite{Kounoike+2005} and~\cite{Asai+2006}: We try all possible 
sequences of flips that create an adjacency in every flip. 
If some such sequence sorts the stack, it is optimal and we are done. Otherwise, we obtain 
a lower bound equal to the number of adjacencies that are needed to be made plus 1 
(here we count pancake $n$ at the bottom of the stack as an adjacency).

In addition, we also use a heuristic to compute 
an upper bound. If the upper bound is equal to the lower bound they give the exact number of flips.

\begin{table}[ht]
\centering
\begin{tabular}{|rrr|rrr|rrr|}
\hline \hline
$n$ & $m$ & $ |\mathbb U_{n}^{m}| $ & $n$ & $m$ & $ |\mathbb U_{n}^{m}| $ & 
$n$ & $m$ & $ |\mathbb U_{n}^{m}| $ \\ \hline
14 & 13 & 30,330,792,508  & 15 & 15 &  310,592,646,490  & 16 & 17 & 756,129,138,051 \\
14 & 14 & 20,584,311,501  & 15 & 16 &   45,016,055,055  & 16 & 18 &       4,646,117 \\
14 & 15 &  2,824,234,896  & 15 & 17 &          339,220  & 17 & 19 &      65,758,725 \\
14 & 16 &         24,974  &    &    &                   &    &    &                 \\

\hline
\end{tabular}
\caption{numbers of stacks of $n$ unburnt pancakes requiring $m$ flips to sort}
\label{tab:fnsizes}
\end{table}

Sizes of the computed sets $\mathbb U_{n}^{m}$ can be found in Table~\ref{tab:fnsizes}.
It was previously known~\cite{HeydariSudb}, that $f(18)\geq 20$ and $f(19)\geq 22$.
No candidate stack of $18$ pancakes needed $21$ flips thus $f(18)=20$. Then $f(19)=22$
because $f(19)\leq f(18)+2 = 22$.

The following modification of this method was also used to compute the values of $g(n)$ up to $n=17$. 
%we first computed the number of flips for each stack of $9$ pancakes by a breadth-first search starting at $I_n$.
Again, $\mathbb C_{n}^{m}$, the set of stacks of $n$ burnt pancakes requiring $m$ flips, is determined
from the set $\cup_{m'=m-2}^{g(n-1)}\mathbb C_{n-1}^{m'}$, but in a slightly different way.
In every stack of $n$ burnt pancakes other than $-I_n$ (which must be treated separately), some two pancakes 
can be joined in two flips~\cite[Theorem 1]{CohenBlum}. We will now show that the two adjacent pancakes 
can be contracted to a single pancake, which decreases the size of the stack. The reverse process is again 
used to determine the stacks of the candidate set, which are then processed by the A* search.

\begin{lemma}
\label{lemma:bcontr}
Let $C$ be a stack of burnt pancakes with a pair $(p_1, p_2)$ of adjacent pancakes 
and let $C'$ be obtained from $C$ by contracting the two adjacent pancakes to a single pancake $p$.
Then $C$ can be sorted in exactly the same number of flips as $C'$.
\end{lemma}
\begin{proof}
If we can sort $C'$ in $m$ steps, we can sort $C$ in $m$ steps as well --- we do the flips 
below the same pancakes as in an optimal sorting sequence for $C'$. Flips in $C'$ below $p$ 
are performed below the lower of $p_1, p_2$ in $C$.

The stack $C'$ can be also obtained from $C$ by removing one of the two adjacent pancakes. 
Then we can sort $C'$ by doing the flips below the same pancakes as in a sorting sequence for $C$.
Flips in $C$ below the removed pancake are performed in $C'$ below the pancake above it.
\end{proof}

%To determine $\mathbb C_{n}^{i}$ for all
%$i \in \{m, m+1,\dots, g(n)\}$ it is thus enough to consider the sets $\mathbb C_{n-1}^{i}$ for
%$i \in \{m-2, m-1,\dots, g(n-1)\}$. In each stack from these sets we try expanding each
%pancake to two then separating the two resulting pancakes by a flip and making one more flip.
%The candidate set composed of $-I_n$ and the stacks that are obtained this way contains all the 
%stacks from the sets $\mathbb C_{n}^{i},~i \in \{m, m+1,\dots\, g(n)\}$. 
%Now it remains to determine the value of $g(C)$ for each stack $C$ in the candidate set.
%As in~\cite{Kounoike+2005} and~\cite{Asai+2006}, this is done using the A* search. 

During the A* search, we compute two lower bounds and take the larger one. One lower bound is computed 
from the formula in Lemma~\ref{lem:blb}. To compute the other lower bound, we try all possible 
sequences of flips that create an adjacency in all but at most two flips. If no such sequence
sorts the stack, we obtain a lower bound equal to the number of adjacencies that are needed
to be made plus 3.

In the stacks visited during the A* search, we can contract a block to a single burnt pancake
thanks to Lemma~\ref{lemma:bcontr}. If, after the contraction of blocks, the stack has at 
most nine pancakes, we look up the exact number of flips in a table previously computed 
by a breadth-first search starting at $I_9$.

\begin{table}[ht]
\centering
\begin{tabular}{|rrr|rrr|rrr|rrr|}
\hline \hline
$n$ & $m$ & $ |\mathbb C_{n}^{m}| $ & $n$ & $m$ & $ |\mathbb C_{n}^{m}| $ & 
$n$ & $m$ & $ |\mathbb C_{n}^{m}| $ & $n$ & $m$ & $ |\mathbb C_{n}^{m}| $\\ \hline
10 & 15 & 22,703,532 & 11 & 17 &  5,928,175  & 12 & 19 &  344,884 & 13 & 21 &   15,675 \\
10 & 16 &    179,828 & 11 & 18 &     10,480  & 12 & 20 &      265 & 13 & 22 &        4 \\
10 & 17 &        523 & 11 & 19 &         36  & 12 & 21 &        1 & 14 & 23 &      122 \\
10 & 18 &          1 &    &    &             &    &    &          & 15 & 25 &        2 \\
\hline
\end{tabular}
\caption{numbers of stacks of $n$ burnt pancakes requiring $m$ flips to sort}
\label{tab:gnsizes}
\end{table}

Sizes of the computed sets $\mathbb C_{n}^{m}$ can be found in Table~\ref{tab:gnsizes}.
No stack of 16 pancakes needs 27 flips thus $g(16)=26$ because $g(-I_{16})=26$. Then $g(17)=28$
because $g(-I_{17})=28$ and $g(17)\leq g(16)+2 = 28$~\cite[Theorem 8]{CohenBlum}.

The stack obtained from $-I_n$ by flipping the topmost pancake is known as $J_n$~\cite{CohenBlum}.
Let $Y_n$ be the stack obtained from $-I_n$ by changing the orientation of the second pancake
from the bottom. The two found stacks of 15 pancakes requiring 25 flips are $J_{15}$ and $Y_{15}$
and they are the first known counterexamples to the Cohen-Blum conjecture which claimed 
that for every $n$, $-I_n$ requires the largest number of flips among all stacks of $n$ pancakes. 
However, no other $J_n$ or $Y_n$ with $n\leq 20$ is a counterexample to the conjecture.

Majority of the computations were done on computers of the CESNET METACentrum grid. Some of the 
computations also took place on computers at the Department of Applied Mathematics of Charles 
University in Prague.

Data and source codes of programs mentioned above can be downloaded from the following webpage:
\url{http://kam.mff.cuni.cz/~cibulka/pancakes}. 

%The two found stacks of 15 pancakes requiring 25 flips are the first known counterexamples to the
%Cohen-Blum conjecture which claimed that for every $n$, $-I_n$ requires the largest 
%number of flips among all stacks of $n$ pancakes. 
%One of the two counterexamples is known as $J_{15}$~\cite{CohenBlum} and is obtained from
%$-I_{15}$ by flipping the topmost pancake. The other is obtained from $-I_{15}$ by changing 
%the orientation of the second pancake from the bottom.

\section{Conclusions}
\label{sec_concl}
Although the two algorithms presented in Sections~\ref{sec:avb}~and~\ref{sec:avu} 
have a good guaranteed average number of flips, experimental results show that both 
of them are often outperformed by the corresponding algorithms of Gates and Papadimitriou.
The average numbers of flips of the two new algorithms are very near to their upper bounds 
calculated in Theorems~\ref{thm:balgo}~and~\ref{thm:ualgo} and the averages for the 
algorithms of Gates and Papadimitriou are in Table~\ref{tab:experiment}.

We will now design one more polynomial-time algorithm for the burnt version, for which no 
guarantee of the average number of flips will be given, but its experimental results are 
close to the lower bound from Theorem~\ref{thm:avgblb}.

Call a sequence of flips, each of which creates an adjacency, a \emph{greedy sequence}. Note
that since we are in the burnt version, there is always at most one possible flip that creates a new 
adjacency. In a random stack the probability that we can join the pancake on top in a single 
flip is $50\%$, therefore starting from a random stack, we can perform a greedy sequence of 
length $\log_2 n$ with probability roughly $1/n$. The idea of the algorithm is, that whenever 
we cannot create an adjacency in a single flip, we try all $n$ possible flips and do
the one that can be followed by the longest greedy sequence.

As in the previous algorithms, two adjacent pancakes are contracted to a single pancake.
Pancakes $1$ and $n$ can create an adjacency ($1$ is viewed as $(n+1)\bmod n$). Therefore 
when the algorithm obtains the stack $(\bsd 1)$ we need at most four more flips.

%\[
%\l( \begin{array}{c}
%	\bsd{i} \\ \bsd{i+1} \\ \vdots \\ \bsd{n} \\ \bsd{1} \\ \vdots \\ \bsd{i-1}
%\end{array} \r)
%\]
%which can be sorted in at most four flips.

In Table~\ref{tab:experiment}, $n$ is the size of a stack, $s_{GP}$ is the average 
number of flips used by the algorithm of Gates and Papadimitriou to sort a randomly 
generated stack of $n$ unburnt pancakes, $s_{GPB}$ is the average number of flips 
used by the algorithm of Gates and Papadimitriou for the burnt version and $s_N$ 
is the average number of flips of the algorithm described in this section.

%For each $n$ several stacks were generated. The number of generated stacks depends on $n$,
%for example for $n=10$, one million stacks were generated, for $n=1000$ there was a thousand 
%of them and for $n=100000$ only ten stacks.

\begin{table}[ht]
\centering
\begin{tabular}{|rrrrrr|}
\hline \hline
$n$    & $s_{GP}$    &  $s_{GPB}$  & $s_N$      & $n+n/\log_2 n$ & stacks generated\\ \hline
10     &  11.129     &  15.383     &  14.935    & 13.010         & 1000000\\
100    &  122.925    &  150.887    & 123.463    & 115.051        &  100000\\
1000   &  1240.949   &  1502.926   & 1127.901   & 1100.343       &   10000\\
10000  &  12408.686  &  15002.212  & 10863.502  & 10752.570      &    1000\\
100000 &  124115.000 &  150063.000 & 106608.900 & 106220.600     &      10\\
1000000& 1241263.600 & 1499875.600 & 1053866.000& 1050171.666    &       5\\
\hline
\end{tabular}
\caption{experimental results of algorithms}
\label{tab:experiment}
\end{table}

The experimental results together with Theorem~\ref{thm:avgblb} support the following conjecture.

\begin{conjecture}
The average number of flips of the optimal algorithm for sorting burnt pancakes satisfies
\[
av_{opt}(n) = n+\Theta\l(\frac{n}{\log n}\r).
\]
\end{conjecture}

\subsection*{Acknowledgements}

I would like to thank Pavel Valtr and Jan Kratochv{\'i}l who led the
seminar under which this article has originated. I would also like to thank
Jan Kyn\v{c}l, Bernard Lidick{\'y}, Radovan \v{S}est{\'a}k and Marek Tesa\v{r}, 
who were participants of the seminar, for their notable comments.
I am also grateful to Pavel Valtr for suggesting to study the average
numbers of flips.

\bibliographystyle{plain}

\end{document}